\documentclass[reqno,11pt,letterpaper]{amsart}

\usepackage{lipsum}
\usepackage{amsmath}
\usepackage{amssymb}
\usepackage{amsthm}
\usepackage{mathrsfs}
\usepackage{accents}
\usepackage{calc}
\usepackage{arydshln}
\usepackage{upgreek}
\usepackage{slashed}
\usepackage{xifthen}
\usepackage{graphicx}
\usepackage{longtable}
\usepackage[inline]{enumitem}

\usepackage{xcolor}
\definecolor{winered}{rgb}{0.6,0,0}
\definecolor{lessblue}{rgb}{0,0,0.7}

\usepackage[pdftex,colorlinks=true,linkcolor=winered,citecolor=lessblue,urlcolor=lessblue,breaklinks=true,bookmarksopen=true]{hyperref}

\makeatletter
\newcommand{\myitem}[3]{\item[#2]\def\@currentlabel{#3}\label{#1}}
\makeatother

\usepackage{titletoc}

\makeatletter

\def\@tocline#1#2#3#4#5#6#7{
\begingroup
  \par
    \parindent\z@ \leftskip#3 \relax \advance\leftskip\@tempdima\relax
                  \rightskip\@pnumwidth plus 4em \parfillskip-\@pnumwidth
    \ifcase #1 
       \vskip 0.6em \hskip 0em 
       \or
       \or \hskip 0em 
       \or \hskip 1em 
    \fi%
    #6
    \nobreak\relax{\leavevmode\leaders\hbox{\,.}\hfill}
    \hbox to\@pnumwidth {\@tocpagenum{#7}}
  \par
\endgroup
}

 \def\l@section{\@tocline{0}{0pt}{0pc}{}{}}

\renewcommand{\tocsection}[3]{%
  \indentlabel{\@ifnotempty{#2}{ 
    \ignorespaces\bfseries{#2. #3}}}
  \indentlabel{\@ifempty{#2}{\ignorespaces\bfseries{#3}}{}} 
    \vspace{1.5pt}}

\renewcommand{\tocsubsection}[3]{%
  \indentlabel{\@ifnotempty{#2}{
    \ignorespaces#2. #3}}
  \indentlabel{\@ifempty{#2}{\ignorespaces #3}{}}
    \vspace{1.5pt}}

\renewcommand{\tocsubsubsection}[3]{%
  \indentlabel{\@ifnotempty{#2}{
    \ignorespaces#2. #3}}
  \indentlabel{\@ifempty{#2}{\ignorespaces #3}{}}
    \vspace{1.5pt}}

\makeatother

\makeatletter
\def\@nomenstarted{0}
\newlength{\@nomenoldtabcolsep}

\newcommand{\nomenstart}
  {%
    \def\@nomenstarted{1}%
    \setlength{\@nomenoldtabcolsep}{\tabcolsep}%
    \setlength{\tabcolsep}{3.5pt}%
    \begin{longtable}{p{0.11\textwidth} p{0.86\textwidth}}
  }

\newcommand{\nomenitem}[2]{%
    \ifcase\@nomenstarted%
      \or 
      \or \\ 
    \fi%
    #1\,{\leavevmode\leaders\hbox{\,.}\hfill} & #2%
    \def\@nomenstarted{2}%
  }%
\newcommand{\nomenend}
  {\\%
      \end{longtable}%
      \setlength{\tabcolsep}{\@nomenoldtabcolsep}%
      \def\@nomenstarted{0}%
  }
\makeatother

\makeatletter
\newcommand{\vast}{\bBigg@{4}}
\newcommand{\Vast}{\bBigg@{5}}
\newcommand{\VAST}[1]{\bBigg@{#1}}
\makeatother

\allowdisplaybreaks

\numberwithin{equation}{section}
\numberwithin{figure}{section}
\newtheorem{thm}{Theorem}[section]

\newtheorem{prop}[thm]{Proposition}
\newtheorem{lemma}[thm]{Lemma}

\newtheorem*{thm*}{Theorem}
\newtheorem*{prop*}{Proposition}
\newtheorem*{cor*}{Corollary}
\newtheorem*{conj*}{Conjecture}

\theoremstyle{definition}

\theoremstyle{remark}
\newtheorem{rmk}[thm]{Remark}

\makeatletter
\newcommand{\fakephantomsection}{%
  \Hy@MakeCurrentHref{\@currenvir.\the\Hy@linkcounter}
  \Hy@raisedlink{\hyper@anchorstart{\@currentHref}\hyper@anchorend}%
  \Hy@GlobalStepCount\Hy@linkcounter%
}
\makeatother

\newcommand{\mc}{\mathcal}

\newcommand{\cC}{\mc C}

\newcommand{\cL}{\mc L}

\newcommand{\cO}{\mc O}

\newcommand{\cU}{\mc U}
\newcommand{\cV}{\mc V}

\newcommand{\ms}{\mathscr}

\newcommand{\sE}{\ms E}

\newcommand{\sK}{\ms K}

\newcommand{\sR}{\ms R}

\newcommand{\R}{\mathbb{R}}

\newcommand{\sfG}{\mathsf{G}}

\newcommand{\supp}{\operatorname{supp}}

\newcommand{\tr}{\operatorname{tr}}

\newcommand{\eps}{\epsilon}

\newcommand{\la}{\langle}

\newcommand{\pa}{\partial}
\newcommand{\dd}{{\mathrm d}}
\newcommand{\ra}{\rangle}

\newcommand{\ul}[1]{\underline{#1}{}}

\newcommand{\pfstep}[1]{$\bullet$\ \underline{\textit{#1}}}

\newcommand{\cp}{{\mathrm{c}}}

\newcommand{\Diff}{\mathrm{Diff}}

\newcommand{\half}{{\tfrac{1}{2}}}

\newcommand{\loc}{{\mathrm{loc}}}
\newcommand{\CI}{\cC^\infty}

\newcommand{\CIc}{\cC^\infty_\cp}

\newcommand{\Ric}{\mathrm{Ric}}
\newcommand{\Ein}{\mathrm{Ein}}

\newcommand{\openbigpmatrix}[1]
  {%
    \def\@bigpmatrixsize{#1}%
    \addtolength{\arraycolsep}{-#1}%
    \begin{pmatrix}%
  }
\newcommand{\closebigpmatrix}
  {%
    \end{pmatrix}%
    \addtolength{\arraycolsep}{\@bigpmatrixsize}%
  }

\newlength{\enummargin}

\newcommand{\usref}[1]{{\upshape\ref{#1}}}

\DeclareGraphicsExtensions{.mps}

\makeatletter
\newcommand*{\fwbw}[1]{\expandafter\@fwbw\csname c@#1\endcsname}
\newcommand*{\@fwbw}[1]{\ifcase #1 \or {\rm fw}\or {\rm bw}\fi}
\AddEnumerateCounter{\fwbw}{\@fwbw}
\makeatother

\begin{document}

\title{The linearized Einstein equations with sources}

\date{\today}

\begin{abstract}
  On vacuum spacetimes of general dimension, we study the linearized Einstein vacuum equations with a spatially compactly supported and (necessarily) divergence-free source. We prove that the vanishing of appropriate charges of the source, defined in terms of Killing vector fields on the spacetime, is necessary and sufficient for solvability within the class of spatially compactly supported metric perturbations. The proof combines classical results by Moncrief with the solvability theory of the linearized constraint equations with control on supports developed by Corvino and Chru\'sciel--Delay.
\end{abstract}


\author{Peter Hintz}
\address{Department of Mathematics, ETH Z\"urich, R\"amistrasse 101, 8092 Z\"urich, Switzerland}
\email{peter.hintz@math.ethz.ch}

\maketitle

\section{Introduction}
\label{SI}

Let $(M,g)$ be a smooth connected globally hyperbolic spacetime, of dimension $n+1$ where $n\geq 1$, which solves the Einstein vacuum equations
\begin{equation}
\label{EqIEinstein}
  \Ein(g)=0,\qquad \Ein(g):=\Ric(g) - \frac12 R_g g.
\end{equation}
Here $\Ric(g)$ and $R_g$ denote the Ricci and scalar curvature, respectively. Let $\Sigma\subset M$ denote a smooth spacelike Cauchy hypersurface; denote its unit normal by $\nu_\Sigma$ and the surface measure by $\dd\sigma$. We study linearized perturbations of $g$ sourced by smooth linearized stress-energy-momentum tensors $f\in\CI_{\rm sc}(M;S^2 T^*M)$ which are \emph{spatially compactly supported} (hence the subscript `sc'): this means that there exists a compact subset $K\subset\Sigma$ so that $\supp f\subset\bigcup_\pm J^\pm(K)$, with $J^\pm(K)\subset M$ denoting the causal future (`$+$'), resp.\ past (`$-$'), of $K$.\footnote{The space $\CI_{\rm sc}(M)$ is well-defined independently of the choice of Cauchy hypersurface.} That is, we shall study the equation
\begin{equation}
\label{EqI}
  D_g\Ein(h) := \frac{\dd}{\dd s}\Ein(g+s h)\Big|_{s=0} = f.
\end{equation}
Recall the second Bianchi identity, which states $\delta_g\Ein(g)=0$ for \emph{all} metrics $g$ (where $\delta_g$ is the negative divergence operator); linearizing this in $g$ and using~\eqref{EqIEinstein} gives
\begin{equation}
\label{EqIDiv}
  \delta_g\bigl(D_g\Ein(h)\bigr) = 0\quad\forall\,h\in\CI(M;S^2 T^*M).
\end{equation}
Therefore, a \emph{necessary} condition for the solvability of~\eqref{EqI} is $\delta_g f=0$. The main result of this note precisely determines the extent to which this condition is also \emph{sufficient}. To state it, denote by $\sK(M,g)\subset\cV(M)=\CI(M;T M)$ the (finite-dimensional) space of Killing vector fields on $(M,g)$.

\begin{thm}[Main theorem, smooth version]
\label{ThmI}
  The map
  \begin{equation}
  \label{EqIIsoMap}
    \CI_{\rm sc}(M;S^2 T^*M)\ni f\mapsto \int_\Sigma f(\nu_\Sigma,X)\,\dd\sigma,\qquad X\in\sK(M,g),
  \end{equation}
  induces a linear isomorphism
  \begin{equation}
  \label{EqIIso}
    \bigl\{ f\in\CI_{\rm sc}(M;S^2 T^*M)\colon \delta_g f=0 \bigr\} \,/\, \bigl\{ D_g\Ein(h) \colon h\in\CI_{\rm sc}(M;S^2 T^*M) \bigr\} \to \sK(M,g)^*.
  \end{equation}
  Here $\sK(M,g)^*=\cL(\sK(M,g),\R)$ is the dual space. This isomorphism is moreover independent of the choice of Cauchy hypersurface $\Sigma$.
\end{thm}

In other words, $D_g\Ein(h)=f\in\CI_{\rm sc}\cap\ker\delta_g$ has a solution $h\in\CI_{\rm sc}$ if and only if all `charges'~\eqref{EqIIsoMap} vanish. (In particular, when $\sK(M,g)=\{0\}$, then $\delta_g f=0$ is sufficient for the solvability of~\eqref{EqI}.) The proof of solvability utilizes the Cauchy problem for a gauge-fixed version of this equation. The Cauchy data must be chosen to satisfy the linearized constraint equations on $\Sigma$ with source $\psi=f(\nu_\Sigma,\cdot)$ while being compactly supported; by results of Corvino \cite{CorvinoScalar} and Chru\'sciel--Delay \cite{ChruscielDelayMapping}, this is possible if and only if $\psi$ is orthogonal to the cokernel of the linearized constraints map, which can be canonically identified with $\sK(M,g)$ by a result of Moncrief \cite{MoncriefLinStabI}. We recall Moncrief's result in Proposition~\ref{PropK} and give a new perspective on its proof. For a finite regularity version of Theorem~\ref{ThmI}, see Theorem~\ref{ThmSobolev}.

In Theorem~\ref{ThmNoObstr}, we show that if $\Sigma$ is noncompact and one drops the support assumptions on $h$, there are no obstructions to solvability beyond $\delta_g f=0$, the reason being that the aforementioned cokernel on the dual space $\sE'(\Sigma)$ of $\CI(\Sigma)$ is trivial. When $M$ and $\Sigma$ have more structure, e.g.\ if they are asymptotically flat, one can make more precise statements regarding the weights at infinity of $f$ and $h$ depending on which (some of) the additional obstructions given by~\eqref{EqIIsoMap} disappear; we shall not discuss this here.

\begin{rmk}[Cosmological constant]
\label{RmkILambda}
  Our arguments go through with purely notational modifications if $\Ein(g)-\Lambda g=0$ where $\Lambda\in\R$ is the cosmological constant. In this case, Theorem~\ref{ThmI} becomes a characterization of those $f\in\CI_{\rm sc}\cap\ker\delta_g$ for which $D_g\Ein(h)-\Lambda h=f$ has a solution $h\in\CI_{\rm sc}$.
\end{rmk}

Our motivation for solving $D_g\Ein(h)=f$ with nontrivial $f$, and understanding the obstructions to solvability, comes from perturbation theory. To give a concrete example, suppose $(M,g)$ is a vacuum spacetime, and we wish to modify it near a timelike geodesic $\gamma$, in local Fermi normal coordinates $(t,x)$ given by $\gamma=\{(t,0)\colon t\in\R\}$, by gluing in a small black hole; in $3+1$ dimensions this could be a Schwarzschild black hole with mass $\eps>0$, given by the metric $-\dd t^2+\dd x^2+\frac{2\eps}{r}(\dd t^2+\dd r^2)+\cO(\eps^2 r^{-2})$ where $r=|x|$. The naive ansatz $g_\eps=g+\chi(r)\frac{2\eps}{r}(\dd t^2+\dd r^2)$ for the modified spacetime metric leads to $\Ein(g_\eps)=\Ein(g)+\eps f+\cO(\eps^2)$ where $f=\cO(1)$ near the gluing region $\supp\chi'\subset\{r>0\}$ (ignoring the singularity of $f$ at $r=0$, which one must deal with separately), and $\delta_g f=0$ by the second Bianchi identity for $g_\eps=g+\cO(\eps)$. One then wishes to eliminate the error $f$ by replacing $g_\eps$ by $g_\eps+\eps h$ where $h$ solves $D_g\Ein(h)=-f$ (while being more regular at $r=0$ than $\frac{2\eps}{r}(\dd t^2+\dd r^2)$). For the details of such a gluing procedure, see \cite{HintzGlueLocI}.

Prior work on solutions of the linearized Einstein equations has largely focused on solutions of the \emph{homogeneous} equation $D_g\Ein(h)=0$. Control of solutions modulo pure gauge solutions (i.e.\ symmetric gradients $\delta_g^*\omega$ of 1-forms $\omega$ on $M$, which always satisfy this equation) is one important problem, especially in the context of stability problems \cite{ReggeWheelerSchwarzschild}. More pertinent to this work is the problem of \emph{linearization stability} as introduced by Fischer--Marsden \cite{FischerMarsdenDeformation,FischerMarsdenLinStabEinstein}, namely whether such an infinitesimal deformation $h$ can be integrated to a nonlinear solution, i.e.\ whether there exists a family $g_s$ of metrics with $\Ein(g_s)=0$ and $h=\frac{\dd}{\dd s} g_s|_{s=0}$. This is the original context of \cite{MoncriefLinStabI}, which shows that linearization stability at $(M,g)$ holds for the Einstein equations if and only if $\sK(M,g)=\{0\}$. The necessary conditions for the existence of $g_s$ whose linearization is $h$ in the presence of nontrivial Killing vector fields were found in \cite{MoncriefLinStabII,ArmsMarsdenMoncriefStructureII,FischerMarsdenMoncriefEinstein}.

In this context, Theorem~\ref{ThmI} gives a necessary and sufficient condition on $f\in\CI_{\rm sc}\cap\ker\delta_g$ so that there exist spacetime metrics $g_s$, $s\in(-1,1)$, with $\Ein(g_s)=s f+\cO(s^2)$. (Absent a general natural physical prescription of further lower order terms beyond $s f$, we do not study the problem of determining when a metric $\tilde g$ near $g$ with, say, $\Ein(\tilde g)=s f$, $s$ small, exists.)  The problem of prescribing the \emph{nonlinear} Ricci curvature tensor of a Lorentzian metric on a given smooth manifold has been studied by DeTurck \cite{DeTurckPrescribedRicciLorentz}; the question of the solvability of the analogue of the constraint equations with source, see \cite[Equations~(3.1)--(3.2)]{DeTurckPrescribedRicciLorentz}, is however not discussed there.

\section{Killing vector fields and the linearized constraint equations}
\label{SK}

We recall classical results by Moncrief \cite{MoncriefLinStabI} and Fischer--Marsden--Moncrief \cite{FischerMarsdenMoncriefEinstein}. For any Lorentzian metric $g$ for which $\Sigma$ is spacelike, the 1-form
\[
  \Ein(g)(\nu_\Sigma,\cdot) \in \CI(\Sigma;T^*_\Sigma M)
\]
depends only on the first and second fundamental forms $\gamma(X,Y)=g(X,Y)$, $k(X,Y)=\la\nabla_X\nu_\Sigma,Y\ra$, $X,Y\in T\Sigma$, of $\Sigma$; this gives rise to the constraints map
\begin{equation}
\label{EqKPEin}
  P(\gamma,k) := \Ein(g)(\nu_\sigma,\cdot) = \Bigl(\frac12\bigl(R_\gamma-|k|_\gamma^2+(\tr_\gamma k)^2\bigr), -\delta_\gamma k-\dd\tr_\gamma k\Bigr).
\end{equation}
We write $D_{\gamma,k}P$ for the linearization of $P$ at $(\gamma,k)$; thus $D_{\gamma,k}P\in\Diff^2(\Sigma;S^2 T^*\Sigma\oplus S^2 T^*\Sigma;T^*_\Sigma M)$. (Here we identify $T_\Sigma M=\ul\R\oplus T\Sigma$ using the orthogonal projections onto $\R\nu_\Sigma$ and $T\Sigma$, and likewise for the cotangent bundles.) Suppose that $\Ein(g)=0$. If $h\in\CI(M;S^2 T^*M)$, and $\dot\gamma,\dot k\in\CI(\Sigma;S^2 T^*\Sigma)$ denote the linearized initial data at $\Sigma$ (i.e.\ the derivatives at $s=0$ of the initial data of $g+s h$), then linearizing~\eqref{EqKPEin} implies
\begin{equation}
\label{EqKPEinLin}
  D_g\Ein(h)(\nu_\sigma,\cdot) = D_{\gamma,k}P(\dot\gamma,\dot k).
\end{equation}

\begin{prop}[Killing vector fields and the linearized constraints map]
\label{PropK}
  (See~\cite{MoncriefLinStabI} and also \cite[Lemma~2.2]{FischerMarsdenMoncriefEinstein}.) Let $(M,g)$ be globally hyperbolic with $\Ein(g)=0$. Then the map $\cV(M)\ni X\mapsto X|_\Sigma\in\CI(\Sigma;T_\Sigma M)$ induces a linear isomorphism
  \begin{equation}
  \label{EqK}
    \sK(M,g) \to \ker (D_{\gamma,k}P)^*.
  \end{equation}
\end{prop}

Our proof of the surjectivity of~\eqref{EqK} is based on a distributional characterization of $\ker (D_{\gamma,k}P)^*$ which seems to not have been noted before; see~\eqref{EqKp0}--\eqref{EqKpi}. Before starting the proof in earnest, we observe that $\sK(M,g)\ni X\mapsto X|_\Sigma\in\CI(\Sigma;T_\Sigma M)$ is injective, in view of the following result:

\begin{lemma}[Unique determination of Killing vector fields]
\label{LemmaKillingUniq}
  (Cf.\ \cite[Lemma~4.1]{MoncriefLinStabI}.) Let $(M,g)$ be a smooth connected $N$-dimensional pseudo-Riemannian manifold. Then the dimension of the space $\sK(M)$ of Killing vector fields on $(M,g)$ satisfies $\dim\sK(M)\leq\frac{N(N+1)}{2}$. Moreover, if $X\in\sK(M)$ vanishes along a hypersurface $\Sigma\subset M$, then $X=0$ on $M$.
\end{lemma}
\begin{proof}
  This follows from the well-known fact that $X$ is uniquely determined by $X(p)\in T_p M$ and $\nabla X(p)\in T_p^* M\otimes T_p M$. (Since the Killing equation imposes $\half N(N+1)$ linearly independent constraints on $X(p),\nabla X(p)$, the space of Killing vector fields has dimension $\leq N+N^2-\half N(N+1)=\half N(N+1)$.) We recall a prolongation argument for the proof of this statement. The Killing equation
  \begin{equation}
  \label{EqKillingUniqEq}
    \la\nabla_V X,W\ra+\la\nabla_W X,V\ra=0,\qquad V,W\in\cV(M),
  \end{equation}
  implies for $Z\in\cV(M)$ the following identity, where we write `$\equiv$' for equality modulo terms involving only $X$ and $\nabla X$ but no higher derivatives:
  \begin{align*}
    &\nabla_Z(\nabla X)(V,W) \equiv \la\nabla_Z\nabla_V X,W\ra = Z\la\nabla_V X,W\ra - \la\nabla_V X,\nabla_Z W\ra \equiv -Z\la\nabla_W X,V\ra \\
    & \quad\equiv -\la\nabla_Z\nabla_W X,V\ra \equiv -\la\nabla_W\nabla_Z X,V\ra \equiv -\nabla_W(\nabla X)(Z,V).
  \end{align*}
  Cyclically permuting the vector fields $(Z,V,W)$ twice more, we obtain $\nabla_Z(\nabla X)(V,W)\equiv-\nabla_Z(\nabla X)(V,W)$, and thus $\nabla_Z(\nabla X)(V,W)$ is an (explicit) expression involving only $X$ and $\nabla X$. Therefore, if $\alpha\colon[0,1]\to M$ is a smooth curve, then there exists a smooth bundle endomorphism $F$ on the restriction of $T M\oplus(T^*M\otimes T M)$ to $\alpha([0,1])$ so that $Z(t)=(X,\nabla X)|_{\alpha(t)}$ satisfies the ODE
  \[
    \frac{D Z}{\dd t} = F(Z).
  \]
  In particular, if $Z(0)=0$, then $Z=0$ on $\alpha([0,1])$, and thus $Z=0$ on $M$ since $M$ is connected.

  To prove the final claim, note that if $V,W$ are vector fields on $M$, then~\eqref{EqKillingUniqEq} gives $V\la X,W\ra+W\la X,V\ra=0$ at $\Sigma$ (where $X=0$). Fix $V$ to be transversal to $\Sigma$ at a point $p\in\Sigma$. For all $W$ which are tangent to $\Sigma$, we have $W\la X,V\ra=W(0)=0$ and thus $V\la X,W\ra=0$, while for $W=V$ we obtain $V\la X,V\ra=0$. Therefore, $\nabla X=0$ at $p$, from where $X=0$ follows from the first part of the proof.
\end{proof}

\begin{proof}[Proof of Proposition~\usref{PropK}]
  \pfstep{Well-definedness; injectivity.} Let $X\in\sK(M,g)$. Given $h\in\CI_{\rm sc}(M;S^2 T^*M)$, write $\dot\gamma,\dot k\in\CIc(\Sigma;S^2 T^*\Sigma)$ for the linearized initial data at $\Sigma$. Since $\delta_g(D_g\Ein(h))=0$, the fact that $X$ is Killing implies that also $D_g\Ein(h)(\cdot,X)$ is divergence-free. Therefore, if $\Sigma_t$, $t\in[-1,1]$, is a smooth family of spacelike hypersurfaces, with unit normal $\nu_{\Sigma_t}$ and surface measure $\dd\sigma_t$, so that $\Sigma_t$ equals $\Sigma$ outside a large compact set, then
  \[
    I(t) := \int_{\Sigma_t} D_g\Ein(h)(\nu_{\Sigma_t},X)\,\dd\sigma_t
  \]
  is independent of $t$ and thus equals
  \[
    I(0) = \int_\Sigma D_g\Ein(h)(\nu_\Sigma,X)\,\dd\sigma = \int_\Sigma \la D_{\gamma,k}P(\dot\gamma,\dot k),X\ra\,\dd\sigma = \int_\Sigma \la(\dot\gamma,\dot k),(D_{\gamma,k}P)^*(X)\ra\,\dd\sigma.
  \]
  Since we can choose $(\dot\gamma,\dot k)\in\CIc(\Sigma;S^2 T^*\Sigma)\oplus\CIc(\Sigma;S^2 T^*\Sigma)$ to be arbitrary at $\Sigma$ and $0$ at $\Sigma_1$, we deduce that $I(0)=0$ for all $\dot\gamma,\dot k$, and therefore $(D_{\gamma,k}P)^*(X)=0$. Together with Lemma~\ref{LemmaKillingUniq}, we conclude that the map~\eqref{EqK} is well-defined and injective. (This argument is taken from~\cite[Lemma~1.5 and Corollary~1.9]{FischerMarsdenMoncriefEinstein}.)

  \pfstep{Surjectivity.} Suppose that $\omega_{(0)}\in\CI(\Sigma;T^*_\Sigma M)$ satisfies
  \begin{equation}
  \label{EqKomega0}
    (D_{\gamma,k}P)^*(\omega_{(0)}^\sharp)=0,
  \end{equation}
  where we use the musical isomorphism for $g$. Write $\delta_g^*$ for the symmetric gradient on 1-forms; this is related to the Lie derivative via $\delta_g^*\omega=\frac12\cL_{\omega^\sharp}g$. We need to show that there exists a solution $\omega\in\CI(M;T^*M)$ of the Killing equation $\delta_g^*\omega=0$ which satisfies $\omega|_\Sigma=\omega_{(0)}$. At $\Sigma$, the requirements $\omega=\omega_{(0)}$ and $0=2(\delta_g^*\omega)(\nu_\Sigma,V)=\nabla_{\nu_\Sigma}\omega(V)+\nabla_V\omega(\nu_\Sigma)$ for $V=\nu_\Sigma$ and $V\in T\Sigma$ uniquely determine the 1-jet $\omega_{(1)}$ of $\omega$ at $\Sigma$. Write $\sfG_g=I-\half g\tr_g$; then we may extend $\omega_{(1)}$ to a 1-form on $M$ by solving the wave equation
  \[
    \delta_g\sfG_g\delta_g^*\omega = 0
  \]
  on $M$, with the Cauchy data of $\omega$ matching $\omega_{(1)}$. Let
  \[
    \pi:=\sfG_g\delta_g^*\omega;
  \]
  thus $\pi|_\Sigma=0\in\CI(\Sigma;S^2 T^*_\Sigma M)$ and $\delta_g\pi=0$ globally. Our aim is to show $\pi=0$. Note that since $\Ein(g)=0$, we have
  \[
    0 = D_g\Ein(\delta_g^*\omega) = \sfG_g D_g\Ric(\sfG_g\pi) = \sfG_g\Bigl(\frac12\Box_g - \delta_g^*\delta_g\sfG_g + \sR_g\Bigr)(\sfG_g\pi) = \frac12\sfG_g(\Box_g+2\sR_g)(\sfG_g\pi)
  \]
  by~\cite[Equation~(2.4)]{GrahamLeeConformalEinstein}, where $\sR_g$ is a 0-th order operator involving the curvature of $g$. In view of this hyperbolic equation for $\sfG_g\pi$, we only need to show that
  \begin{equation}
  \label{EqKpiJet}
    \text{the 1-jet of $\sfG_g\pi$ (or, equivalently, that of $\pi$) vanishes at $\Sigma$}
  \end{equation}
  to conclude the proof. This is where the equation satisfied by $\omega_{(0)}$ will enter.

  To wit,~\eqref{EqKomega0} is equivalent to
  \[
    0 = \int_\Sigma \la D_g\Ein(h), \nu_\Sigma\otimes_s\omega_{(0)}^\sharp \ra\,\dd\sigma = \la D_g\Ein(h), (\nu_\Sigma\otimes_s\omega_{(0)}^\sharp)\delta(\Sigma)\ra
  \]
  for all $h\in\CI_{\rm sc}(M;S^2 T^*M)$, where we use the distributional pairing on $M$ on the right. This is further equivalent to
  \begin{equation}
  \label{EqKp0}
    0 = (D_g\Ein)^*\bigl( (\nu_\Sigma^\flat\otimes_s\omega_{(0)})\delta(\Sigma) \bigr) = D_g\Ein( \dd H\otimes_s\omega_{(0)} )
  \end{equation}
  where $H\in L^1_\loc(M)$ is $1$, resp.\ $0$ in the past, resp.\ future of $\Sigma$. Finally, since $\dd H\otimes_s\omega_{(0)}=\delta_g^*(H\omega)-H\delta_g^*\omega$ and $D_g\Ein\circ\delta_g^*=0$ as well as $D_g\Ein=\sfG_g\circ D_g\Ric\circ\sfG_g$, we find that~\eqref{EqKomega0} is equivalent to the distributional equation
  \begin{equation}
  \label{EqKpi}
    [D_g\Ric,H]\pi=0,\qquad \pi=\sfG_g\delta_g^*\omega.
  \end{equation}
  Expanding this and using $\pi|_\Sigma=0$ gives
  \[
    0 = [\Box_g,H]\pi - 2[\delta_g^*\delta_g\sfG_g,H]\pi = [\Box_g,H]\pi - 2[\delta_g^*,H](\delta_g\sfG_g\pi) - 2\delta_g^*\bigl([\delta_g\sfG_g,H]\pi\bigr),
  \]
  with the final term on the right being zero since $[\delta_g\sfG_g,H]\pi=-(\sfG_g\pi)(\nabla H,\cdot)=0$ (this being the product of a $\delta$-distribution at $\Sigma$ with a tensor vanishing at $\Sigma$). In abstract index notation, we thus have
  \[
    -g^{\mu\nu}\bigl(H_{;\mu\nu}\pi_{\kappa\lambda} + H_{;\mu}\pi_{\kappa\lambda;\nu} + H_{;\nu}\pi_{\kappa\lambda;\mu}\bigr) - H_{;\kappa}(\delta_g\sfG_g\pi)_\lambda - H_{;\lambda}(\delta_g\sfG_g\pi)_\kappa = 0.
  \]
  Since $H_{;\mu}\pi_{\kappa\lambda}=0$, the first two terms in parentheses cancel. Using further that $\delta_g\sfG_g\pi=\delta_g\pi+\half\dd\tr_g\pi=\half\dd\tr_g\pi$ gives $-H_{;\mu}\pi_{\kappa\lambda;}{}^\mu - \frac12\bigl(H_{;\kappa}\pi_\mu{}^\mu{}_{;\lambda} + H_{;\lambda}\pi_\mu{}^\mu{}_{;\kappa}\bigr) = 0$. Choosing coordinates $z^0,\ldots,z^n$ so that $\Sigma=\{z^0=0\}$ and $\dd z^0$ is a unit covector, this gives
  \[
    0 = \pi_{\kappa\lambda;0} - \frac12\bigl(\delta_{0\kappa}\pi_\mu{}^\mu{}_{;\lambda} + \delta_{0\lambda}\pi_\mu{}^\mu{}_{;\kappa}\bigr) = \pi_{\kappa\lambda;0} + \delta_{0\kappa}\delta_{0\lambda}(\pi_{0 0;0}-\pi_m{}^m{}_{;0})
  \]
  at $\Sigma$, where $m$ runs from $1$ to $n$; here we use that $\pi$ and therefore also its spatial covariant derivatives vanish at $\Sigma$. For $\kappa=i\geq 1$, $\lambda\geq 0$, this gives $\pi_{i\lambda;0}=0$. For $\kappa=\lambda=0$ then, we get $\pi_{0 0;0}=0$. This shows~\eqref{EqKpiJet} and finishes the proof.
\end{proof}

\section{Proof of the main result; variants}
\label{SPf}

The independence of the map~\eqref{EqIIsoMap} of the choice of Cauchy hypersurface $\Sigma$ in~\eqref{EqIIsoMap} follows from the divergence theorem, since $\delta_g f=0$ and $X\in\sK(M,g)$ implies that $f(X,\cdot)$ is divergence-free. The vanishing of $\int_\Sigma f(\nu_\Sigma,X)\,\dd\sigma$ for $f=D_g\Ein(h)$ was shown in the first step of the proof of Proposition~\ref{PropK}. Thus,~\eqref{EqIIso} is well-defined.

For the proof of Theorem~\ref{ThmI}, we require the following (well-known) result:

\begin{thm}[Solvability of the linearized constraints]
\label{ThmPfConstraints}
  Let $(\Sigma,\gamma)$ be a smooth connected Riemannian manifold, and let $k\in\CI(\Sigma;S^2 T^*\Sigma)$. Suppose $\psi\in\CIc(\Sigma;\ul\R\oplus T^*\Sigma)$ satisfies $\la\psi,f^*\ra=0$ for all $f^*\in\ker(D_{\gamma,k}P)^*\subset\CI(\Sigma;\ul\R\oplus T^*\Sigma)$. Then there exist $\dot\gamma,\dot k\in\CIc(\Sigma;S^2 T^*\Sigma)$ so that
  \[
    D_{\gamma,k}P(\dot\gamma,\dot k) = \psi.
  \]
\end{thm}
\begin{proof}
  The existence of regular solutions with controlled support is due to Corvino \cite{CorvinoScalar}, with the existence of smooth solutions proved by Chru\'sciel--Delay \cite{ChruscielDelayMapping}; see \cite{DelayCompact} for general results of this flavor. The key point is that $(D_{\gamma,k}P)^*$ is an overdetermined (Douglis--Nirenberg) elliptic partial differential operator for which a priori estimates hold on spaces of tensors defined on a smoothly bounded open precompact subset $\cU\subset\Sigma$, containing $\supp\psi$, which allow for exponential growth at $\pa\cU$. By duality, this gives the solvability of $D_{\gamma,k}P(\dot\gamma,\dot k)=\psi$, provided $\psi$ is orthogonal to the cokernel, with $\dot\gamma,\dot k$ exponentially decaying at $\pa\cU$, whence their extension by $0$ to $\Sigma\setminus\cU$ furnishes the desired solution.
\end{proof}

\begin{proof}[Proof of Theorem~\usref{ThmI}]
  \pfstep{Injectivity of~\eqref{EqIIso}.} Suppose $f\in\CI_{\rm sc}(M;S^2 T^*M)$ satisfies $\delta_g f=0$ and $\int_\Sigma f(\nu_\Sigma,X)\,\dd\sigma=0$ for all $X\in\sK(M,g)$. By Proposition~\ref{PropK}, $f(\nu_\Sigma,\cdot)\in\CIc(\Sigma;T^*_\Sigma M)$ is orthogonal to $\ker_{\CI(\Sigma;T^*_\Sigma M)} (D_{\gamma,k}P)^*$.
  
  Choose $\dot\gamma,\dot k\in\CIc(\Sigma;S^2 T^*\Sigma)$ according to Theorem~\ref{ThmPfConstraints} with $\psi=f(\nu_\Sigma,\cdot)$. Pick then $h_0,h_1\in\CIc(\Sigma;S^2 T^*_\Sigma M)$ so that $\dot\gamma,\dot k$ are the linearized initial data corresponding to any metric perturbation $\tilde h\in\CI_{\rm sc}(M;S^2 T^*M)$ with Cauchy data $h_0=\tilde h|_\Sigma$, $h_1=\nabla_{\nu_\Sigma}\tilde h$. Let $\theta\in\CIc(M;T^*M)$ be any extension of $\theta_{(0)}:=\delta_g\sfG_g\tilde h|_\Sigma\in\CIc(\Sigma;T^*_\Sigma M)$. We then solve the initial value problem for the gauge-fixed linearized Einstein vacuum equation
  \begin{equation}
  \label{EqPfGaugeFixed}
    D_g\Ric(h) + \delta_g^*(\delta_g\sfG_g h-\theta) = \sfG_g f,\qquad (h|_\Sigma,\nabla_{\nu_\Sigma}h) = (h_0,h_1);
  \end{equation}
  the solution $h$ satisfies $h\in\CI_{\rm sc}(M;S^2 T^*M)$ by finite speed of propagation. If we set $\eta:=\delta_g\sfG_g h-\theta$, then we have
  \begin{equation}
  \label{EqPfEta0}
    \eta|_\Sigma = (\delta_g\sfG_g\tilde h-\theta)|_\Sigma = 0
  \end{equation}
  by definition of $\theta$. Moreover, applying $\sfG_g$ to~\eqref{EqPfGaugeFixed} and recalling~\eqref{EqKPEinLin}, we obtain
  \[
    (\sfG_g\delta_g^*\eta)(\nu_\Sigma,\cdot) = f(\nu_\Sigma,\cdot) - D_g\Ein(h)(\nu_\Sigma,\cdot) = 0.
  \]
  Together with~\eqref{EqPfEta0}, this implies $\nabla_{\nu_\Sigma}\eta=0$. Finally, applying $\delta_g\sfG_g$ to~\eqref{EqPfGaugeFixed} and using the linearized second Bianchi identity together with $\delta_g f=0$ gives the hyperbolic equation
  \[
    \delta_g\sfG_g\delta_g^*\eta = 0
  \]
  for $\eta$, whence $\eta=0$ on $M$ and thus $D_g\Ric(h)=\sfG_g f$. Therefore, $f=D_g\Ein(h)$ projects to $0$ in the quotient space~\eqref{EqIIso}.

  \pfstep{Surjectivity of~\eqref{EqIIso}.} Given $\lambda\in\sK(M,g)^*$, there exists a 1-form $\psi\in\CIc(\Sigma;T^*_\Sigma M)$ with $\la\psi,X|_\Sigma\ra=\lambda(X)$ for all $X\in\sK(M,g)$. Indeed, it suffices to arrange this for $X$ in a basis of $\sK(M,g)$, in which case it follows from the linear independence of the restrictions of the basis elements to $\Sigma$ (a consequence of Lemma~\ref{LemmaKillingUniq}). Pick then any $f_{(0)}\in\CIc(\Sigma;S^2 T^*_\Sigma M)$ with $f_{(0)}(\nu_\Sigma,\cdot)=\psi$, and let $\tilde f\in\CIc(M;S^2 T^*M)$ denote any compactly supported symmetric 2-tensor with $\tilde f=f_{(0)}$ at $\Sigma$. We claim that there exists $\omega\in\CI_{\rm sc}(M;T^*M)$ with $(\omega|_\Sigma,\nabla_{\nu_\Sigma}\omega)=0$ so that
  \[
    f := \tilde f + \sfG_g\delta_g^*\omega \in \ker\delta_g.
  \]
  Indeed, we simply solve the initial value problem
  \[
    \delta_g\sfG_g\delta_g^*\omega = -\delta_g\tilde f \in \CIc(M;S^2 T^*M),\qquad
    (\omega|_\Sigma,\nabla_{\nu_\Sigma}\omega) = (0,0).
  \]
  Since $\sfG_g\delta_g^*\omega$ vanishes at $\Sigma$ for such $\omega$, we have $\int_\Sigma f(\nu_\Sigma,X)\,\dd\sigma=\int_\Sigma f_{(0)}(\nu_\Sigma,X)\,\dd\sigma=\int_\Sigma\la\psi,X\ra\,\dd\sigma=\lambda(X)$ for all $X\in\sK(M,g)$ still, and the surjectivity of~\eqref{EqIIso} follows. This completes the proof of Theorem~\ref{ThmI}.
\end{proof}

\begin{rmk}[Equivalence to solving the linearized constraints]
\label{RmkEquivConstraints}
  Not only did Theorem~\ref{ThmPfConstraints} play a key role in the argument; one can conversely \emph{deduce} Theorem~\ref{ThmPfConstraints} from Theorem~\ref{ThmI}. Indeed, given $\psi$ satisfying the assumptions of Theorem~\ref{ThmPfConstraints}, one constructs $f\in\ker\delta_g\cap\CI_{\rm sc}(M;S^2 T^*M)$ with $f(\nu_\Sigma,\cdot)=\psi$ as in the last part of the above proof. Theorem~\ref{ThmI} produces $h\in\CI_{\rm sc}(M;S^2 T^*M)$ with $D_g\Ein(h)=f$; evaluating this on $\Sigma$ at $(\nu_\Sigma,\cdot)$ gives $D_{\gamma,k}P(\dot\gamma,\dot k)=\psi$ where $\dot\gamma,\dot k\in\CIc(\Sigma;S^2 T^*\Sigma)$ are the linearized initial data induced by $h$.
\end{rmk}

\begin{thm}[Finite regularity]
\label{ThmSobolev}
  Let $s\geq 0$. Suppose that $f\in H^s_{\rm sc}(M;S^2 T^*M)$ (i.e.\ $f$ lies in $H^s_\loc(M;S^2 T^*M)$ and has spatially compact support) satisfies $\delta_g f=0$, and $\int_\Sigma f(\nu_\Sigma,X)\,\dd\sigma=0$ for all $X\in\sK(M,g)$.\footnote{In coordinates $(z^0,\ldots,z^n)$ in which $\dd z^0$ is a unit conormal at $\Sigma=\{z^0=0\}$, we have $f_{0\mu;0}=f_{j\mu;}{}^j$. Therefore $f(\nu_\Sigma,\cdot)\in H_\cp^{s-\frac12}(\Sigma;T^*_\Sigma M)$ is well-defined for $s>-\frac12$.} Then there exists $h\in H^{s+1}_{\rm sc}(M;S^2 T^*M)$ with $D_g\Ein(h)=f$.
\end{thm}

Together with the surjectivity of~\eqref{EqIIso}, this shows that~\eqref{EqIIsoMap} induces an isomorphism
\[
  \{ f\in H^s_{\rm sc} \colon \delta_g f=0 \} / \{ D_g\Ein(h) \colon h\in H^{s+1}_{\rm sc}\ \text{with}\ D_g\Ein(h)\in H^s_{\rm sc} \} \to \sK(M,g)^*.
\]

\begin{proof}[Proof of Theorem~\usref{ThmSobolev}]
  Foliating a neighborhood of $\Sigma$ by level sets $\Sigma_s:=t^{-1}(s)$ (which we identify with $\Sigma$), $s\in(-1,1)$, of a time function $t\in\CI(M)$, we have, a fortiori, $f\in L^1_\loc((-1,1),H^s_\cp(\Sigma;S^2 T^*_\Sigma M))$, whence $f|_{t=\tau}\in H^s_\cp$ for almost all $\tau$; fix such a $\tau\in(-1,1)$. Using a finite regularity version of Theorem~\ref{ThmPfConstraints}, one can then find $\dot\gamma\in H^{s+2}_\cp(\Sigma;S^2 T^*\Sigma)$, $\dot k\in H^{s+1}_\cp(\Sigma;S^2 T^*\Sigma)$, so that $D_{\gamma,k}P(\dot\gamma,\dot k)=f(\nu_{\Sigma_\tau},\cdot)$ at $\Sigma_\tau$, and thus Cauchy data $h_0\in H^{s+2}_\cp$, $h_1\in H^{s+1}_\cp$ inducing $\dot\gamma,\dot k$; we can then compute $\theta_{(0)}\in H^{s+1}_\cp(\Sigma;T^*_\Sigma M)$ as in the earlier proof, extend it to be constant in $t$ followed by cutting it off to a neighborhood of $t=\tau$ to obtain $\theta\in\CIc((-1,1);H^{s+1}_\cp)$. The equation~\eqref{EqPfGaugeFixed} we then solve for $h$ has spatially compactly supported forcing in $L^1_\loc H^s_\cp$ and thus a spatially compactly supported solution of class $\cC^0 H_\cp^{s+1}\cap\cC^1 H_\cp^s$. The arguments leading to $D_g\Ein(h)=f$ then go through without changes since $\eta=\delta_g\sfG_g h-\theta$ and $\sfG_g\delta_g^*\eta(\nu_{\Sigma_\tau},\cdot)$ vanish at $\Sigma_\tau$. Finally, $f\in H_{\rm sc}^s$ follows by standard hyperbolic theory \cite[\S23.2]{HormanderAnalysisPDE3}.
\end{proof}

\begin{thm}[Solvability without support conditions]
\label{ThmNoObstr}
  Suppose the Cauchy hypersurface $\Sigma$ inside the globally hyperbolic spacetime $(M,g)$ with $\Ein(g)=0$ is noncompact. Let
  \[
    f\in\CI(M;S^2 T^*M),\qquad \delta_g f=0.
  \]
  Then there exists $h\in\CI(M;S^2 T^*M)$ with $D_g\Ein(h)=f$.
\end{thm}

Thus, there are no obstructions to solvability if we drop the support assumptions on $h$, and one can even drop the support assumptions on the source term $f$. (Even for $f\in\CI_{\rm sc}(M;S^2 T^*M)$, the solution $h$ produced below is typically large at infinity.)

\begin{proof}[Proof of Theorem~\usref{ThmNoObstr}]
  We need to prove the existence of $\dot\gamma,\dot k\in\CI(\Sigma;S^2 T^*\Sigma)$ so that $D_{\gamma,k}P(\dot\gamma,\dot k)=f(\nu_\Sigma,\cdot)$. Once this is done, the arguments using the gauge-fixed equation~\eqref{EqPfGaugeFixed} apply verbatim to produce a solution of $D_g\Ein(f)=h$. The existence of $(\dot\gamma,\dot k)$ follows from \cite[Theorem~37.2]{TrevesTVS} once we show that the adjoint operator
  \[
    (D_{\gamma,k}P)^*\colon\sE'(\Sigma;\ul\R\oplus T^*\Sigma)\to\sE'(\Sigma;S^2 T^*\Sigma\oplus S^2 T^*\Sigma)
  \]
  has trivial kernel and weak* closed range. The injectivity follows from the fact that an element of the kernel vanishes outside some compact set and is thus identically zero, as follows from the relationship of $\ker(D_{\gamma,k}P)^*$ with Killing vector fields on $(M,g)$ and Lemma~\ref{LemmaKillingUniq} (or directly on the level of $(D_{\gamma,k}P)^*$ via \cite[Lemma~4.3]{HintzGlueID}). The weak* closed range property follows, via the characterization \cite[Theorem~37.1]{TrevesTVS}, from the fact that $(D_{\gamma,k}P)^*$ is overdetermined (Douglis--Nirenberg) elliptic, which indeed implies that for all $s\in\R$ and compact $K\subset\Sigma$, the operator $(D_{\gamma,k}P)^*\colon\dot H^{s+2}(K)\oplus\dot H^{s+1}(K;T^*_K\Sigma)\to\dot H^s(K;S^2 T^*_K\Sigma\oplus S^2 T^*_K\Sigma)$ has closed range. (Here $\dot H^s(K)$ denotes the subspace of $H^s_\loc(\Sigma)$ consisting of distributions with support in $K$.)
\end{proof}

\bibliographystyle{alphaurl}


\begin{thebibliography}{AMM82}

\bibitem[AMM82]{ArmsMarsdenMoncriefStructureII}
Judith~M. Arms, Jerrold~E. Marsden, and Vincent Moncrief.
\newblock The structure of the space of solutions of {E}instein's equations
  {II}: several {K}illing fields and the {E}instein--{Y}ang--{M}ills equations.
\newblock {\em Annals of Physics}, 144(1):81--106, 1982.
\newblock \href {https://doi.org/10.1016/0003-4916(82)90105-1}
  {\path{doi:10.1016/0003-4916(82)90105-1}}.

\bibitem[CD03]{ChruscielDelayMapping}
Piotr~T. Chru{\'s}ciel and Erwann Delay.
\newblock {On mapping properties of the general relativistic constraints
  operator in weighted function spaces, with applications}.
\newblock {\em Mem. Soc. Math. France}, 94:1--103, 2003.

\bibitem[Cor00]{CorvinoScalar}
Justin Corvino.
\newblock Scalar curvature deformation and a gluing construction for the
  {E}instein constraint equations.
\newblock {\em Comm. Math. Phys.}, 214(1):137--189, 2000.

\bibitem[Del12]{DelayCompact}
Erwann Delay.
\newblock {S}mooth compactly supported solutions of some underdetermined
  elliptic {PDE}, with gluing applications.
\newblock {\em Communications in Partial Differential Equations},
  37(10):1689--1716, 2012.

\bibitem[DeT83]{DeTurckPrescribedRicciLorentz}
Dennis~M. DeTurck.
\newblock The {C}auchy problem for {L}orentz metrics with prescribed {R}icci
  curvature.
\newblock {\em Compositio Mathematica}, 48(3):327--349, 1983.

\bibitem[FM73]{FischerMarsdenLinStabEinstein}
Arthur~E. Fischer and Jerrold~E. Marsden.
\newblock {L}inearization stability of the {E}instein equations.
\newblock {\em Bulletin of the American Mathematical Society}, 79(5):997--1003,
  1973.

\bibitem[FM75]{FischerMarsdenDeformation}
Arthur~E. Fischer and Jerrold~E. Marsden.
\newblock Deformations of the scalar curvature.
\newblock {\em Duke Mathematical Journal}, 42(3):519--547, 1975.

\bibitem[FMM80]{FischerMarsdenMoncriefEinstein}
Arthur~E. Fischer, Jerrold~E. Marsden, and Vincent Moncrief.
\newblock {The structure of the space of solutions of Einstein's equations. I.
  One Killing field}.
\newblock {\em Annales de l'IHP Physique th{\'e}orique}, 33(2):147--194, 1980.

\bibitem[GL91]{GrahamLeeConformalEinstein}
C.~Robin Graham and John~M. Lee.
\newblock Einstein metrics with prescribed conformal infinity on the ball.
\newblock {\em Adv. Math.}, 87(2):186--225, 1991.

\bibitem[Hin22]{HintzGlueID}
Peter Hintz.
\newblock Gluing small black holes into initial data sets.
\newblock {\em Preprint, arXiv:2210.13960}, 2022.

\bibitem[Hin23]{HintzGlueLocI}
Peter Hintz.
\newblock Gluing small black holes along timelike geodesics {I}: formal
  solution.
\newblock {\em Preprint}, 2023.

\bibitem[H{\"o}r07]{HormanderAnalysisPDE3}
Lars H{\"o}rmander.
\newblock {\em The analysis of linear partial differential operators. {III}}.
\newblock Classics in Mathematics. Springer, Berlin, 2007.

\bibitem[Mon75]{MoncriefLinStabI}
Vincent Moncrief.
\newblock Spacetime symmetries and linearization stability of the {E}instein
  equations. {I}.
\newblock {\em Journal of Mathematical Physics}, 16(3):493--498, 1975.
\newblock \href {https://doi.org/10.1063/1.522572}
  {\path{doi:10.1063/1.522572}}.

\bibitem[Mon76]{MoncriefLinStabII}
Vincent Moncrief.
\newblock Space–time symmetries and linearization stability of the {E}instein
  equations. {II}.
\newblock {\em Journal of Mathematical Physics}, 17(10):1893--1902, 1976.
\newblock \href {https://doi.org/10.1063/1.522814}
  {\path{doi:10.1063/1.522814}}.

\bibitem[RW57]{ReggeWheelerSchwarzschild}
Tullio Regge and John~A. Wheeler.
\newblock {S}tability of a {S}chwarzschild {S}ingularity.
\newblock {\em Phys. Rev.}, 108:1063--1069, Nov 1957.

\bibitem[Tr{\`e}67]{TrevesTVS}
Fran\c{c}ois Tr{\`e}ves.
\newblock {\em Topological vector spaces, distributions and kernels}.
\newblock Academic Press, New York-London, 1967.

\end{thebibliography}

\end{document}